\documentclass{article}

\usepackage{amssymb}

\newtheorem{thm}{Theorem}
\newtheorem{cor}[thm]{Corollary}

\newtheorem{lemma}[thm]{Lemma}
\newtheorem{definition}[thm]{Definition}

\newenvironment{proof}{\noindent\bf{Proof.}\rm}{\hfill$\blacksquare$\bigskip}

\begin{document}

\title{Interchanging distance and capacity in probabilistic mappings}

\author{Reid Andersen\thanks{Microsoft Live Labs,
One Microsoft Way, Redmond, WA 98052. {\tt reidan@microsoft.com}.}
\and Uriel Feige\thanks{Department of Computer Science and Applied
Mathematics, Weizmann Institute, Rehovot 76100, Israel. {\tt
uriel.feige@weizmann.ac.il}. Supported in part by The Israel
Science Foundation (grant No. 873/08).}}

\maketitle

\begin{abstract}
    Harald R\"{a}cke [STOC 2008] described a new method to obtain
hierarchical decompositions of networks in a way that minimizes the
congestion. R\"{a}cke's approach is based on an equivalence that he
discovered between minimizing congestion and minimizing stretch
(in a certain setting). Here we present R\"{a}cke's equivalence in an
abstract setting that is more general than the one described in
R\"{a}cke's work, and clarifies the power of R\"{a}cke's result. In
addition, we present a related (but different) equivalence that
was developed by Yuval Emek [ESA 2009] and is only known to apply
to planar graphs.
\end{abstract}

\section{Introduction}

In this manuscript we present results of a manuscript of Harald
R\"{a}cke titled ``Optimal hierarchical decompositions for congestion
minimization in networks"~\cite{Racke}. Our presentation is more
modular than the original presentation of R\"{a}cke in that it
separates the existential aspects of R\"{a}cke's result from the
algorithmic aspects. The existential results are presented in a
more abstract setting that allows the reader to appreciate the
generality of R\"{a}cke's result. Our presentation is also more
careful not to lose on the tightness of the parameters (e.g., not
to give away constant multiplicative factors). For slides of a
talk based on this manuscript see~\cite{FeigeSlides}.

Our manuscript is organized as follows. In
Section~\ref{sec:bisection} we discuss the optimization problem of
min-bisection. Achieving an improved approximation ratio for this
problem is one of the results of~\cite{Racke}, and  we use this
problem as a motivation for the main results that follows. In
Section~\ref{sec:existence} we present in an abstract setting what
we view as R\"{a}cke's main result, namely, an equivalence between two
types of probabilistic embeddings, one concerned with faithfully
representing distances and the other with faithfully representing
capacities. In Section~\ref{sec:algorithms} we briefly discuss
algorithmic versions of the existential result. In
Section~\ref{sec:applications} we show how the machinery developed
leads to an approximation algorithm for min-bisection. In
Section~\ref{sec:planar} we present results related to the main
theme of this manuscript, but that do not appear in~\cite{Racke}.
These results concern an equivalence between deterministic
embeddings in planar graphs, and were first developed and used by
Yuval Emek in~\cite{emek}.

\section{Min-bisection}
\label{sec:bisection}

In the min-bisection problem, the input is a graph with an even
number $n$ of vertices. In the weighted version of the problem,
edges have arbitrary nonnegative weights, whereas in the
unweighted version, the weight of every edge is~1. A bisection of
the graph is a partition the set of vertices into two sets of
equal size. The width of a bisection is the total weight of the
edges that are cut (an edge is cut if its endpoints are on
different sides of the partition). Min-bisection asks for a
bisection of minimum width. This problem is NP-hard.

One line of research dealing with the NP-hardness of min-bisection
offers a bi-criteria approximation. Namely, it is concerned with
developing algorithms that produce a partition of the graph into
nearly equal parts (rather than exactly equal parts), such that
the width of the partition is not much larger than the width of
the minimum bisection. The methodology used by these algorithms
was developed in a sequence of papers and currently allows one to
efficiently find a near bisection (e.g., each of the two parts has
at least one third of the vertices) whose width is within a
multiplicative factor of $O(\sqrt{\log n})$ of the width of the
minimum bisection. The methodology used by these papers is related
to the theme of the current manuscript. It also uses an interplay
between distance and capacity. We briefly explain this interplay,
and refer the readers to~\cite{LR,LLR,ARV} for more details.

For simplicity, assume that the input graph is a complete graph,by
replacing non-edges by edges of weight 0. View the weight of an
edge in a min-bisection problem as specifying its capacity. The
width of a bisection is the total capacity of its associated cut.
The first phase of the bicriteria approximation algorithm involves
solving some linear program LP (in~\cite{LR,LLR}) or semidefinite
program SDP (in~\cite{ARV}). The output of this mathematical
program can be thought of as a fractional cut in the following
sense: edges are assigned lengths, and the longer an edge is, a
larger fraction of it belongs to the cut. Naturally, for the
fractional solution to have small value, the LP (or SDP) will try
to assign short lengths to edges of high capacity. (This is a
theme that will reappear in the proof of
Theorem~\ref{thm:equivalentM}.) Thereafter, this fractional
solution is rounded to give a near bisection of width not much
larger than the value of the fractional solution. The rounding
procedure is more likely to cut the long edges than the short
edges. Or equivalently, two vertices of short distance from each
other (where distance is measured by sum of edge lengths along
shortest path) are likely to fall in the same side of the
partition. Hence to find a near bisection of small capacity, the
bicriteria approximation algorithms introduce an intermediate
notion of distance, and carefully choose (as a solution to the LP
or SDP) a distance function that interacts well with the
capacities of the edges. In fact, there is a formal connection
between the distortion of this distance function in comparison to
$\ell_1$ distances and the approximation ratio (in terms of
minimizing capacity) that one gets from this methodology.
(See~\cite{LLR} for an exact statement.)

To move from a bicriteria approximation to a true approximation
(in which the output is a true bisection, and approximation is
only in the sense that the width is not necessarily optimal), it
appears unavoidable that one should use in some way dynamic
programming. Consider the task of determining whether a graph has
a bisection of width 0. Such a bisection exists if and only if
there a set of connected components of the graph whose total
size is exactly $n/2$. Determining whether this is the case
amounts to solving a subset sum problem (with sizes of connected
components serving as input to the subset sum problem), and the
only algorithm known to solve subset sum (in time polynomial in
the numbers involved) is dynamic programming. In~\cite{FK}, the
techniques used in the bicreteria approximation were combined with
a dynamic programming approach to produce a true approximation for
min-bisection with approximation ratio $O((\log n)^{3/2})$
(obtained as the bicriteria approximation ratio times $O(\log
n)$). Here we shall describe R\"{a}cke's approach that gives a better
approximation ratio, $O(\log n)$.

Trees (and more generally, graphs of bounded treewidth, though
this subject is beyond the scope of the current manuscript) form a
family of graphs on which many NP-hard problems can be solved
using dynamic programming. In particular, min bisection can be
solved in polynomial time on trees. This suggests the following
plan for approximating min-bisection on general graphs, which is
presented here using terms that are suggestive but have not been
defined yet. First, find a ``low distortion embedding" of the
input graph into a tree. Then solve min bisection optimally on the
tree, using dynamic programming. The solution will induce a
bisection on the original graph, and the approximation ratio will
be bounded by the distortion of the embedding into the tree.

The plan as described above has certain drawbacks. One is that the
distortion when embedding a general graph into a tree might be
very large (e.g., the distance distortion for an $n$-cycle). This
problem has been addressed in a satisfactory way in previous
work~\cite{AKPW,bartal}. Rather than embed the graph in one tree,
one finds a probabilistic embedding into a family of (dominating)
trees (the requirement that the trees be dominating is a technical
requirement that will be touched upon in
Section~\ref{sec:applications}), and considers average distortion
(averaged over all trees). When the objective function is linear
(as in the case of min bisection), the probabilistic notion of
embedding suffices. Hence the modified plan is as follows. Find a
low distortion probability distribution over embeddings of the
input graph into (dominating) trees. Then solve min bisection
optimally on each of these trees. Each of these solutions will
induce a bisection on the original graph, and for the best of them
the approximation ratio will be bounded by the {\em average}
distortion of the probabilistic embedding into trees.

The known probabilistic embeddings of graphs into trees are
tailored to minimize average distortion, where the aspect that is
being distorted is distance between vertices. However, for the
intended application of min bisection, the aspect that interests
us is the capacity of cuts rather than distances. Hence R\"{a}cke's
approach for approximating min-bisection is as described above,
but with the distinction that the distortion of embeddings is
measured with respect to capacity rather than distance.

To implement this approach, one needs to design probabilistic
embeddings with low capacity distortion. Here is an informal
statement of R\"{a}cke's result in this respect.

\begin{thm}
\label{thm:RackeInformal} For every graph on $n$ vertices there is
a probability distribution of embeddings into (dominating) trees
with $O(\log n)$ average distortion of the capacity. Moreover,
these embeddings can be found in polynomial time.
\end{thm}

This theorem is analogous to known theorems regarding the
distortion of distances in probabilistic embeddings~\cite{FRT}.
R\"{a}cke's proof of Theorem~\ref{thm:RackeInformal} is by a reduction
between these two types of embeddings. The existence of such a
reduction may not be unexpected (as an afterthought), because as
we have seen in the bicriteria approximation algorithms, there are
certain correspondences between capacity and distance.

As a direct consequence of the theorem above and the modified plan
for approximating min bisection, one obtains an $O(\log n)$
approximation for min bisection. This will be discussed in more
detail in Section~\ref{sec:applications}.

\section{An abstract setting}
\label{sec:existence}

In this section we present an abstract setting that as a special
case will lead to the existential component of
Theorem~\ref{thm:RackeInformal}.

\subsection{Definitions}
\label{sec:definitions}

Let $E$ be a set (of edges) and $\cal{P}$ a collection of nonempty multisets of $E$
(that we call paths). A mapping $M:E \longrightarrow \cal{P}$ maps to every edge $i
\in E$ a path $P \in \cal{P}$. It will be convenient to represent a mapping by a
matrix $M$, where $M_{ij}$ counts the number of times the edge $j$ lies on the path
$M(i)$.

{\bf Spanning tree example.} $E$ is a set of edges of a connected
graph $G$. Consider a spanning tree $T$ of $G$, and let $\cal{P}$
be the set of simple paths in $T$. Then there is a natural mapping
$M$ that maps every edge $(i,j)\in E$ to the set of edges that
form the unique simple path between vertices $i$ and $j$ in $T$.

{\bf Tree embedding example.} $E$ is a set of edges of a connected
graph $G$. Consider an arbitrary tree $T$ defined over the same
set of vertices as $G$ (edges of $T$ need not be edges of $G$). As
in the spanning tree example above, there is a natural mapping
from edges of $G$ to paths in $T$. However, this is not a mapping
in the sense defined above, because edges of $T$ are not
necessarily edges of $G$. To remedy this situation, we represent
each edge $(i,j)$ of $T$ by a set of edges that form a simple path
between $i$ and $j$ in $G$. This representation is not unique
(there may be many simple paths between $i$ and $j$ in $G$), and
hence some convention is used to specify one such path uniquely
(for example, one may take a shortest path, breaking ties
arbitrarily). Hence now each edge of $T$ corresponds to a set of
edges in $G$, and each simple path in $T$ corresponds to a
collection of several such sets. Now $T$ can be the mapping that
maps each edge $(i,j)$ in $G$ to a multiset of edges of $G$ that
is obtained by joining together (counting multiplicities) the sets
of edges of $G$ that form the paths that correspond to the edges
of $T$ that lie along the simple path connecting $i$ and $j$ in
$T$.

{\bf Graph embedding example.} This is a generalization of the
tree embedding example. $E$ is a set of edges of a connected graph
$G$ on $n$, and $H$ is an arbitrary different graph defined on the
same set of vertices. Now an edge $(i,j)\in E$ is mapped to some
path (using a convention such as that of taking a shortest path)
that connects $i$ and $j$ in $H$, and this path in $H$ is
represented as a multiset of edges in $G$ (as in the tree
embedding example). This defines a mapping $M$. Natural
alternative versions of the mapping reduce the multiset to a set,
either by removing multiplicities of edges in the multiset, or by
more extensive processing (e.g., if the multiset corresponded to a
nonsimple path in $G$ that contains cycles, these cycles may
possibly be removed).

{\bf Hypergraph example.} $E$ is the set of hyperedges of a
connected 3-uniform hypergraph $H$. $T$ is a spanning tree of the
hypergraph $H$, in the sense that it is defined on the same set of
vertices as $H$, and every edge $(i,j)$ of $T$ is labeled by some
hyperedge of $H$ that contains vertices $i$ and $j$. Then a
hyperedge ${i,j,k}$ can be mapped to the set of hyperedges that
label the set of edges along the paths that connect $i,j$ and $k$
in $T$.

We note that the tree embedding example is essentially the setting
in R\"{a}cke's work~\cite{Racke}. However the spanning tree example
suffices in order to illustrate the main ideas in R\"{a}cke's work.

Let $\cal{M}$ be a family of {\em admissible} mappings. A
probabilistic mapping between $E$ and $\cal{M}$ is a probability
distribution over mappings $M \in \cal{M}$. That is, with every $M
\in \cal{M}$ we associate some $\lambda_M \ge 0$, with $\sum_M
\lambda_M = 1$.

We shall consider probabilistic mappings in two different contexts.

\begin{definition}
\label{def:distanceM}
{\bf (Distance mapping).} Every edge $i$ has a positive length $\ell_i$ associated
with it. We let $dist_M(i)$ denote the length of the path $M(i)$, namely $dist_M(i)
= \sum_j M_{ij}\ell_j$. The {\em stretch} of an edge $i$ is
$\frac{dist_M(i)}{\ell_i}$. The average stretch of an edge in a probabilistic
mapping is the weighted average (weighted according to $\lambda_M$) of the stretches
of the edge. The stretch of a probabilistic mapping is the maximum over all edges of
their average stretches.
\end{definition}

The stretch of a particular edge may be smaller than~1. However,
the stretch of the shortest edge will always be at
least~1.

Probabilistic distance mappings were considered
in~\cite{AKPW,EEST,ABN} for the spanning tree example, and
in~\cite{bartal,FRT} for the tree embedding example.

\begin{definition}
\label{def:congestionM}
{\bf (Capacity mapping).} Every edge $i$ has a positive capacity $c_i$ associated
with it. We let $load_M(j)$ denote the sum (with multiplicities) of capacities of
edges whose path under $M$ contains $j$. Namely, $load_M(j) = \sum_i M_{ij}c_i$.
The {\em congestion} of an edge $j$ is $\frac{load_M(j)}{c_j}$. The average
congestion of an edge in a probabilistic mapping is the weighted average (weighted
according to $\lambda_M$) of the congestions of the edge. The congestion of a
probabilistic mapping is the maximum over all edges of their average congestions.
\end{definition}

The congestion of an edge may be smaller than~1. However, the sum
of all capacities in $M(E)$ is at least as large as the sum of all
capacities in $E$, implying that the congestion of a probabilistic
mapping is always at least~1.

For concreteness, let us present the notions of distance, stretch,
load and congestion as applied to the spanning tree example.
Consider a connected graph $G$ in which every edge $e = (i,j)$ has
a positive length $\ell_e$ and a positive capacity $c_e$. Consider
an arbitrary spanning tree $T$ of $G$, and the mapping from $G$ to
$T$ described in the spanning tree example above. Then the
distance of edge $e$ is the sum of length of edges along the
unique simple path that connects vertices $i$ and $j$ in $T$. The
stretch of $e$ is then the ratio between this distance and
$\ell_e$. The load on edge $e$ is 0 if $e$ is not part of the
spanning tree. However, if $e$ is part of the spanning tree, the
load is computed as follows. Removing $e$, the tree $T$ decomposes
into two trees, one containing vertex $i$ (that we call $T_i$) and
the other containing vertex $j$ (that we call $T_j$). The load of
$e$ is the sum of capacities of all edges (including $e$ itself)
that have one endpoint in $T_i$ and the other in $T_j$. The
congestion of $e$ is the ratio between the load and $c_e$.

\subsection{Probabilistic mappings as zero-sum games}
\label{sec:games}

We shall use the following standard consequence of the minimax
theorem for zero sum games (as in~\cite{AKPW}).

\begin{lemma}
\label{lem:distanceM} For every $\rho \ge 1$ and every family of admissible mappings $\cal{M}$, there is a probabilistic mapping with stretch at
most $\rho$ if and only if for every nonnegative coefficients
$\alpha_i$, there is a mapping $M \in \cal{M}$ such that $\sum
\alpha_i \frac{dist_M(i)}{\ell_i} \le \rho \sum \alpha_i$.
\end{lemma}

\begin{proof}
Consider a zero sum game in which the player MAP chooses an
admissible mapping $M$, and the player EDGE chooses an edge $i$.
The value of the game for for EDGE is the stretch of $i$ in the
mapping, and hence EDGE wishes to maximize the stretch whereas MAP
wishes to minimize it. A probabilistic mapping is a randomized
strategy for MAP. Choosing nonnegative coefficients $\alpha_i$
(and scaling them so that $\sum \alpha_i = 1$) is a randomized
strategy for EDGE.

{\em The ``only if" direction.} If there is no randomized strategy
for MAP forcing an expected value at most $\rho$, then the minimax
theorem implies that there must be a randomized strategy for EDGE
that enforces an expected value more than $\rho$, regardless of
which mapping player MAP chooses to play.

{\em The ``if" direction.} If there is no randomized strategy for
EDGE forcing an expected value larger than $\rho$, then the
minimax theorem implies that there must be a randomized strategy
for MAP that enforces an expected value of at most $\rho$,
regardless of which edge player EDGE chooses to play.
\end{proof}

\begin{lemma}
\label{lem:congestionM}
For every $\rho \ge 1$ and every family of admissible mappings $\cal{M}$, there is a
probabilistic mapping with congestion at most $\rho$ if and only if for every
nonnegative coefficients $\beta_i$, there is a mapping $M \in \cal{M}$ such that
$\sum \beta_i \frac{load_M(i)}{c_i} \le \rho \sum \beta_i$.
\end{lemma}

The proof of Lemma~\ref{lem:congestionM} is similar to the proof of
Lemma~\ref{lem:distanceM}, and hence is omitted.

\subsection{Main result}

\begin{thm}
\label{thm:equivalentM}
For every $\rho \ge 1$ and every family of admissible mappings $\cal{M}$, the following two
statements are equivalent:

\begin{enumerate}

\item For every collection of lengths $\ell_i$ there is a probabilistic mapping with
stretch at most $\rho$.

\item For every collection of capacities $c_j$ there is a probabilistic mapping with
congestion at most $\rho$.

\end{enumerate}

\end{thm}

\begin{proof}
We first prove that item~2 implies item~1.

Assume that there is a probabilistic mapping from $E$ using
$\cal{M}$ with congestion at most $\rho$. By
Lemma~\ref{lem:congestionM}, for every nonnegative coefficients
$\beta_j$, there is a mapping $M \in \cal{M}$ such that $\sum
\beta_j \frac{load_M(j)}{c_j} \le \rho \sum \beta_i$. Hence, using
the notation from Definition~\ref{def:congestionM}, for every
nonnegative coefficients $\beta_j$ we have a mapping satisfying:

\begin{equation}
\label{eq:congestionM}
\sum_{j,i \in E} \beta_j M_{ij} \frac{c_i}{c_j} \le \rho \sum_{j\in E} \beta_j
\end{equation}

We need to prove that there is a probabilistic mapping from $E$ using $\cal{M}$ with
stretch at most $\rho$. By Lemma~\ref{lem:distanceM}, it suffices to prove that for
every nonnegative coefficients $\alpha_i$, there is a mapping $M \in \cal{M}$ such
that $\sum \alpha_i \frac{dist_M(i)}{\ell_i} \le \rho \sum \alpha_i$. Hence, using
the notation from Definition~\ref{def:distanceM}, for every nonnegative coefficients
$\alpha_i$ we need to find a mapping satisfying:

\begin{equation}
\label{eq:distanceM}
\sum_{i,j \in E} \alpha_i M_{ij} \frac{\ell_j}{\ell_i} \le \rho \sum_{i\in E} \alpha_i
\end{equation}

Choosing $\beta_j = \alpha_j$ and $c_i = \alpha_i/\ell_i$ (and likewise, $c_j =
\alpha_j/\ell_j$) and substituting in inequality~(\ref{eq:congestionM}), we obtain
inequality~(\ref{eq:distanceM}).

The proof that item~1 implies item~2 is similar, choosing $\alpha_i = \beta_i$ and
$\ell_i = \beta_i/c_i$, and then inequality~(\ref{eq:distanceM}) becomes
inequality~(\ref{eq:congestionM}).
\end{proof}

Observe that the proof of Theorem~\ref{thm:equivalentM} does not assume that entries
of matrices $M \in \cal{M}$ are nonnegative integers (except for the issue that
$\rho$ is stated as a quantity of value at least~1). Neither does it assume that
distances or capacities are nonnegative (though they cannot be~0, since the
expressions for stretch and congestion involve divisions by distances or
capacities).

\subsection{Simultaneous stretch and congestion bounds}
\label{sec:simultaneous}

Let $\cal{M}$ be a family of admissible mappings for which there is a probabilistic mapping
with stretch at most $\rho$. Hence by Theorem~\ref{thm:equivalentM}, there is also a
(possibly different) probabilistic mapping with congestion at most $\rho$. However, this does not imply
that there is a probabilistic mapping which
simultaneously achieves stretch at most $\rho$ and congestion at most $\rho$.

Consider the following example. The edges $E$ are the set of edges
of the following graph $G$. $G$ has two special vertices denoted
by $s$ and $t$. There are $\sqrt{n}$ vertex disjoint paths between
$s$ and $t$, each with $\sqrt{n}$ edges. In addition, there is the
edge $(s,t)$. Hence altogether $E$ contains $n + 1$ edges. The
family $\cal{M}$ of admissible mappings is the canonical family corresponding to the set of all spanning trees of $G$, as described in Section~\ref{sec:definitions}; given an edge of $E$ and a spanning tree $T$ of $G$,
the edge is mapped to the unique path in $T$ joining its
endpoints.

Let $\ell$ be an arbitrary length function on $E$. The following
is a probabilistic mapping into $\cal{M}$ of stretch at most~3.
Let $P$ be the shortest path in $G$ between $s$ and $t$ (breaking
ties arbitrarily). It may be either the edge $(s,t)$ or one of the
$\sqrt{n}$ paths. For the probabilistic mapping, we choose a random spanning tree as follows.  All edges of $P$ are contained in
the spanning tree. In addition, from every path $P' \not= P$,
exactly one edge is deleted, with probability proportional to the
length of the edge.

Let us analyze the expected stretch of the above probabilistic
mapping. For edges along the path $P$, the stretch is~1. Consider
now an edge of length $\ell$ in a path $P' \not= P$ of length $L$.
With probability $(L - \ell)/L < 1$ the edge remains in the random
spanning tree, keeping its original length. With probability
$\ell/L$ the edge is not in the spanning tree, and then it is
mapped to a path of length at most $2L - \ell < 2L$ (we used here
the fact that the length of $P$ is at most $L$). Hence the
expected stretch is smaller than $1 +
2L\frac{\ell}{L}\frac{1}{\ell} \le 3$.

Let $c$ be an arbitrary capacity function on $E$. The following is
a probabilistic mapping into $\cal{M}$ of congestion at most~3.
For every path $P_j$ between $s$ and $t$ (including the edge
$(s,t)$ as one of the paths), let $e_j$ denote the edge of minimum
capacity on this path, and let $c_j$ be its capacity. Choose a
random spanning tree that contains all edges except the $e_j$
edges, and exactly one of the $e_j$ edges, chosen with probability
proportional to its capacity.

Let us analyze the expected congestion of the above probabilistic
mapping. Consider an arbitrary edge $e$ on path $P_j$. Its
capacity $c$ is at least $c_j$. The load that it suffers is at
most its own capacity $c$, plus perhaps the capacity $c_j \le c$
of edge $e_j$, plus with probability $c_j/(\sum_i c_i) \le
c/(\sum_i c_i)$ a capacity of $\sum_{i \not= j} c_i \le \sum_i
c_i$, which in expectation contributes at most $c$. Hence
altogether, its expected load is at most $3c$.

The two probabilistic mappings that we designed (one for stretch,
one for congestion) are very different. In particular, the sizes
of the supports are $n^{\sqrt{n}}$ for the stretch case and
$\sqrt{n} + 1$ for the congestion case. We now show that if in the
graph $G$ all lengths and all capacities are equal to~1 (the
unweighted case), every probabilistic mapping must have either
stretch or congestion at least $\sqrt{n}/2$. To achieve stretch
less than $\sqrt{n}/2$, the edge $(s,t)$ must belong to the random
spanning tree with probability at least $1/2$. However, whenever
$(s,t)$ is in the spanning tree, then from every other path
connecting $s$ and $t$ one edge needs to be removed, contributing
a total load of $\sqrt{n}$ to the edge $(s,t)$. It is also
interesting to observe that a random spanning tree (chosen
uniformly at random from all spanning trees) contains the edge
$(s,t)$ with probability exactly 1/2 (this is because the
effective resistance between $s$ and $t$ is 1/2, details omitted).
Hence the uniform distribution over spanning trees is
simultaneously bad (distortion at least $\sqrt{n}/2$) both for
stretch and for congestion.

We have seen that the probabilistic mappings that achieve low
congestion may be very different than those that achieve low
stretch. However, as we shall see in Section~\ref{sec:planar}, for
the special case considered here (that of spanning trees in planar
graphs), there are additional connections between stretch and
congestion, based on planar duality.

\section{Algorithmic aspects}
\label{sec:algorithms}

In Section~\ref{sec:existence} our discussion was concerned only
with the existence of probabilistic mappings. Here we shall
discuss how such mappings can be found efficiently. As we have
seen in Section~\ref{sec:games}, and using the notation of
Lemma~\ref{lem:distanceM}, the problem of finding a probabilistic
mapping with smallest distortion can be cast as a problem of
finding an optimal mixed strategy for the player MAP, in a zero
sum game between the players MAP and EDGE. We shall briefly review
the known results concerning the computation of optimal mixed
strategies in zero sum games, and how they can be applied in our
setting. Hence in a sense, this section is independent of
Section~\ref{sec:existence}.

\subsection{An LP formulation of zero sum games}

It is well known that the value of a zero sum game is a solution
to a linear program (LP), and that linear programming duality in
this case implies the minimax theorem.

Consider a game matrix $A$ with $r$ rows (the pure strategies for
MAP) and $c$ columns (the pure strategies for EDGE), in which
entry $A_{ij}$ contains the payoff for EDGE if MAP plays row $i$
and EDGE plays column $j$. Map wishes to select a mixed strategy
that minimizes the expected payoff, whereas EDGE wishes to select
a mixed strategy that maximizes this payoff.

With each row $i$ we can associate a variable $x_i$ that denotes
the probability with which row $i$ is played in MAP's mixed
strategy. Then the linear program is to {\em minimize} $\rho$
subject to:

\begin{itemize}

\item
$\sum_i A_{ij}x_i \le \rho$ for all columns $j$,

\item
$\sum x_i = 1$,

\item
$x_i \ge 0$ for all $i$.

\end{itemize}

An immediate consequence of this LP formulation is that an optimal
solution to the LP can be found in time polynomial in the size of
the game matrix $A$ (e.g., using the Ellipsoid algorithm).
However, in our context of probabilistic mappings, it will often
be the case that the size of $A$ is not polynomial in the
parameters of interest. For example, when mapping a graph $G$ into
a distribution over spanning trees, the parameter of interest is
typically $n$, the number of vertices in the graph. The number of
edges (and hence number of columns in $A$) is at most $n^2$ and
hence polynomially bounded in $n$. However, the number of mappings
(number of spanning trees in $G$) might be of the order of
$n^{\Omega(n)}$, which is not polynomial in $n$. Hence if one is
interested in algorithms with running time polynomial in $n$, one
cannot even write down the matrix $A$ explicitly (though the graph
$G$ serves as an implicit representation of $A$).

Though the case that will interest us most is when there are
superpolynomially many row strategies and polynomially many column
strategies, let us discuss first the case when there are
polynomially many row strategies and superpolynomially many column
strategies.

We note that in the discussions that follow we shall assume that
all payoffs are rational numbers with numerators and denominators
represented by a number of bits that is polynomial in a parameter
of interest (such as the smallest of the two dimensions of $A$).

\subsection{Superpolynomially many column strategies}
\label{sec:expcolumn}

As we cannot afford to write the matrix $A$ explicitly, we need to
assume some other mechanism for accessing the strategies of the
column player. Typically, this is viewed abstractly as {\em
oracle} access. A natural oracle model is the following:

{\bf Best response oracle}. Given a mixed strategy for the row
player, the oracle provides a pure strategy for the column player
of highest expected payoff (together with the corresponding column
of $A$).

If a best response oracle is available, one may still run the
ellipsoid algorithm (with the best response oracle serving as a
separation oracle) and obtain an optimal solution to the LP (and
hence an optimal mixed strategy for the game). See~\cite{GLS} for
more details on this approach.

\subsection{Superpolynomially many row strategies}
\label{sec:exprow}

Here we address the main case of interest, when there are
superpolynomially many row strategies, but only polynomially many
column strategies. One issue that has to be dealt with is whether
an optimal mixed strategy for the row player can be represented at
all in polynomial space, given that potentially it requires
specifying probabilities for superpolynomially many strategies.
Luckily, the answer is positive. Mixed strategies are solutions to
linear programs, and linear programs have {\em basic feasible
solutions} in which the number of nonzero variables does not
exceed the number of constraints (omitting the nonnegativity
constraints $x_i \ge 0$). Hence there is an optimal mixed strategy
for the row player whose support (the number of pure strategies
that have positive probability of being played) is not larger than
the number of columns.

To access the pure strategies of the row player, let us assume
here that we have a best response oracle for the row player. Now a
standard approach is to consider the dual of the LP, which
corresponds to finding an optimal mixed strategy for the column
player. By analogy to Section~\ref{sec:expcolumn}, one can use the
ellipsoid algorithm to find an optimal mixed strategy for the
column player.

An optimal solution to the dual LP has the same value as an
optimal solution to the primal LP, but is not by itself a solution
to the primal LP (and hence, we still did not find an optimal
mixed strategy for the row player). However, there are certain
ways of leveraging the ability to solve the dual LP and using it
so as to also solve the primal LP. One such approach employs an
exploration phase that finds polynomially many linearly
independent constraints of the dual that are tight at the optimal
dual solution, and then find an optimal primal solution that is
supported only on the primal variables that correspond to these
dual constraints. Details are omitted here (but presumably appear
in~\cite{GLS}).

\subsection{Weaker oracle models and faster algorithms}
\label{sec:weaker}

In the games that interest us, typically there are polynomially
many columns (the column player is EDGE who can play an edge in
the graph) and superpolynomially many rows (the row player Map has
exponentially many mappings to choose from). In this respect, we
are in the setting of Section~\ref{sec:exprow}. However, we might
not have a best response oracle representation of MAP. Instead we
shall often have a weaker kind of oracle.

{\bf $\delta$-response oracle}.  Given a mixed strategy for the
column player, the oracle provides a pure strategy for the row
player (together with the corresponding row of $A$) that limits
the expected payoff (to column player) to at most $\delta$.

Let $\rho$ be the true minimax value of the game. Then the value
of $\delta$ for a $\delta$-response oracle must be at least
$\rho$, but in general might be much larger than $\rho$. If this
is the only form of access to the pure strategies of the row
player, finding an optimal mixed strategy for the row player
becomes hopeless. Hence the goal is no longer to find the optimal
mixed strategy for the row player, but rather to find a mixed
strategy that limits the expected payoff of the column player to
at most $\delta$ (plus low order terms). We sketch here an
approach of Freund and Schapire~\cite{FS} that can be used. It is
based on the use of {\em regret minimizing} algorithms.

Consider an iterative process in which in each round, the column
player selects a mixed strategy, the row player selects in
response a $\delta$-response (pure) strategy, and the column
player collects the expected payoff of his mixed strategy against
that pure strategy. If the column player is using a regret
minimizing online algorithm in order to select his mixed
strategies (such as using a multiplicative weight update rule),
then after polynomially many rounds (say, $t$), his payoff (which
can be at most $\delta t$) is guaranteed to approach (up to low
order terms) the total payoff that the best fixed column pure
strategy can achieve against the actual sequence of pure
strategies played by the row player. This means that if the row
player plays the mixed strategy of choosing one of the $t$ rounds
at random and playing the row strategy that was played in this
round, no pure column strategy has expected payoff significantly
larger than $\delta$. For more details on this subject, the reader
is referred to~\cite{FS}, or to surveys such as~\cite{BM}
or~\cite{AHK}.

\subsection{Implementation for probabilistic mappings}
\label{sec:implementation}

Consider a zero sum game as in the setting of
Lemma~\ref{def:distanceM}. EDGE has polynomially many strategies,
whereas MAP potentially has exponentially many strategies. A best
response oracle for MAP needs to be able to find a best response
for MAP against any given mixed strategy of EDGE. In many
contexts, finding the best response is NP-hard. However, for the
intended applications, often a $\delta$-response suffices,
provided that one can guarantee that $\delta$ is not too large.
Indeed, given coefficients $\alpha_i$ of the edges of the input
graph, one can find a spanning tree with average stretch
$\tilde{O}(\log n)$~\cite{ABN} (the $\tilde{O}$ notation hides
some lower order multiplicative terms), and a tree embedding of
average stretch $O(\log n)$~\cite{FRT}. This in combination with
an algorithmic framework similar to that outlined in
Section~\ref{sec:weaker} gives probabilistic embeddings with
stretch $\tilde{O}(\log n)$ and $O(\log n)$ respectively.

The above results in combination with
Theorem~\ref{thm:equivalentM} imply that there is also a
probabilistic mapping into spanning trees with congestion
$\tilde{O}(\log n)$ and a probabilistic mapping into (arbitrary)
trees with congestion $O(\log n)$. To actually find such a mapping
algorithmically, one needs to find a mixed strategy for the player
MAP in the corresponding zero sum game. The algorithmic framework
of Section~\ref{sec:weaker} shows that this can be done if we can
implement a $\delta$-response oracle for MAP. We have already seen
that such an oracle can be implemented for distance mapping, but
now need to do so for capacity mappings. Luckily, the proof of
Theorem~\ref{thm:equivalentM} can be used for this purpose. It
shows how to transform any $\delta$-response query to a capacity
mapping oracle into a $\delta$-response query to a distance
mapping oracle. This establishes the algorithmic aspect of
Theorem~\ref{thm:RackeInformal}.  We remark that the resulting
probabilistic mappings have support size polynomial in $n$.

\section{Applications}
\label{sec:applications}

R\"{a}cke describes several applications to his results, with
oblivious routing being a prominent example. Here we concentrate
only on one of the applications, that of min-bisection that served
as our motivating example.

Let $G$ be a connected graph on an even number $n$ of vertices, in
which edges have nonnegative capacities. One wishes to find a
bisection of minimum width (total capacity of edges with endpoints
in different sides of the bipartization). We present here a
polynomial time algorithm with approximation ratio $\tilde{O}(\log
n)$.

Consider an arbitrary spanning tree $T$ of $G$. Every edge $e =
(i,j)$ of $T$ partitions the vertices of $G$ into two sets that we
call $T_i$ and $T_j$. Define the load $load_T(e)$ of edge $e$ to
be the sum of capacities of edges of $G$ with one endpoint in
$T_i$ and the other endpoint in $T_j$. (This is consistent with
Section~\ref{sec:definitions}.)

Consider now an arbitrary bipartization $B$ of $G$. Let $E_T(B)$
be the set of edges of $T$ that have endpoints in different sides
of the bipartization. Then the width of the bipartization is at
most $\sum_{e\in E_T(B)} load_T(e)$. (The load terms count every
edge of $G$ cut by the bipartition at least once and perhaps
multiple times, and possibly also count edges of $G$ not in the
bipartization.) This is the {\em domination} property that we were
referring to in Section~\ref{sec:bisection}.

By Theorem~\ref{thm:RackeInformal}, whose proof is summarized in
Section~\ref{sec:applications}, one can find in polynomial time a
distribution over spanning trees of $G$ such that for every edge
of $G$, its expected congestion (over choice of random spanning
tree) is at most $\delta = \tilde{O}(\log n)$.

Consider an optimal bisection in $G$, and let $b$ denote its
width. For each edge cut by the bisection, its expected congestion
over the probabilistic mapping into spanning trees is at most
$\delta$. Summing over all edges cut by the bisection and taking a
weighted average over all spanning trees in the probabilistic
mapping, we obtain that at least in one such tree $T$, the width
of this bipartization (with respect to the load in that tree) is
at most $\delta b$.

The above discussion gives the following algorithm for finding a
bisection of small width in a graph $G$ whose minimum bisection
has width $b$.

\begin{enumerate}

\item Find a probabilistic mapping into spanning trees with
congestion at most $\delta$. (By the discussion above this step
takes polynomial time, and $\delta$ can be taken to be
$\tilde{O}(\log n)$. Furthermore, the set of spanning trees in the support of the probabilistic mapping has size polynomial in $n$.)

\item In each spanning tree, find an optimal bisection (with
respect to the load) using dynamic programming. This takes
polynomial time. Moreover, by the discussion above, in at least
one tree the bisection found will have width at most $\delta b$.

\item Of all the bisections found (one per spanning tree), take
the one that in $G$ has smallest width. By the domination
property, its width is at most $\delta b$.

\end{enumerate}

The approximation ratio that we presented above for min-bisection
is $\delta = \tilde{O}(\log n)$, rather than $O(\log n)$ as was
done by R\"{a}cke. To get the $O(\log n)$ approximation, instead of
probabilistic mappings into spanning trees one simply uses
probabilistic mappings into (arbitrary but dominating) trees. Then
one can plug in the bounds of~\cite{FRT} rather than the somewhat
weaker bounds of~\cite{ABN} and obtain the desired approximation
ratio. Details omitted.

\section{Spanning trees in planar graphs} \label{sec:planar}

In Section~\ref{sec:simultaneous} we saw that for distributions
over spanning trees of planar graphs, the distributions achieving
low stretch are very different from those achieving low
congestion. In this section we present an interesting connection
between low stretch and low congestion for spanning trees in
planar graphs. A similar connection was observed independently
(and apparently, before our work) by Yuval Emek~\cite{emek}.

The family of graphs that we shall consider is that of 2-connected
planar multigraphs. Specifically, the graphs need to be planar,
connected, with no cut edge (an edge whose removal disconnects the
graph), and parallel edges are allowed. In the context of spanning
trees, restricting graphs to be 2-connected is not a significant
restriction, because disconnected graphs do not have spanning
trees, and every cut edge belongs to every spanning tree and hence
does not contribute to the complexity of the problem. The reason
why we allow parallel edges is so that the notion of a dual of a
planar graph will always be defined. From every set of parallel
edges, a spanning tree may contain at most one edge.

Every planar graph can be embedded in the plane with no
intersecting edges. In fact, several algorithms are known to
produce such embeddings in linear time. This embedding might not
be unique, in which case we fix one planar embedding arbitrarily.
Given a planar embedding, the {\em dual} graph is obtained by
considering every face of the embedding (including the outer face)
to be a vertex of the dual graph, and every edge of the embedding
corresponds to an edge of the dual graph that connects the two
vertices that correspond to the two faces that the edge separates.
(The fact that the graph has no cut edges insures that the dual
has no self loops. Two faces that share more than one edge give in
the dual parallel edges.) The dual graph is planar and the planar
embedding that we associate with it is the one naturally obtained
by the above construction. Under this planar embedding, the dual
of the dual is the primal graph. Cycles in the primal graph
correspond to cuts in the dual graph and vice versa. It is well
known and easy to see that given a spanning tree in the primal
graph, the edges not in the spanning tree form a spanning tree in
the dual graph. (This also gives Euler's formula that $|V| - 1 +
|F| - 1 = |E|$.)

Consider now a length function on the edges of the primal graph.
Given a spanning tree of the primal graph, for every spanning tree
edge its stretch is~1, and for every other edge its stretch is
determined by the length of the fundamental cycle that it closes
with the spanning tree edges. In the dual graph, let the capacity
of an edge be equal to the length of the corresponding edge in the
primal graph. Consider the dual spanning tree. The congestion of
edges not on the dual spanning tree is~0 (one less than their
stretch in the primal). The load of an edge on the dual spanning
tree is precisely the sum of capacities of the corresponding
fundamental cycle in the primal graph, and hence the congestion is
exactly one more than the stretch in the primal.

The above deterministic correspondence has the following
probabilistic corollary.

\begin{cor}
\label{cor:planar} Consider an arbitrary 2-connected planar graph
$G$ and its planar dual $\bar{G}$. Assume that edges in $G$ have
nonnegative lengths whereas edges in $\bar{G}$ have nonnegative
capacities, and moreover, the capacity of an edge in $\bar{G}$ is
equal to the length of the corresponding edge in $G$. Then for
every probabilistic mapping of $G$ into spanning trees with
stretch $\rho$, the same distribution over the dual trees forms
probabilistic mapping of $\bar{G}$ into spanning trees with
congestion at most $\rho+1$. Likewise, for every probabilistic
mapping of $\bar{G}$ into spanning trees with congestion $\rho$,
the same distribution over the dual trees forms a probabilistic
mapping of $G$ into spanning trees with stretch at most $\rho+1$.
\end{cor}

\subsection*{Acknowledgements}

We thank Satyen Kale for helpful discussions on the subjects of
this manuscript.

\end{document}